\documentclass[11pt]{article}
\usepackage{fullpage}
\usepackage{amsmath}  %for \align
\usepackage{amssymb, amsthm, algorithm, algorithmic,thmtools,thm-restate,enumerate,verbatim,bm}
\usepackage[usenames,dvipsnames,svgnames,table]{xcolor}
\definecolor{darkgreen}{rgb}{0.0,0,0.9}
\usepackage[pagebackref,letterpaper=true,colorlinks=true,pdfpagemode=none,citecolor=OliveGreen,linkcolor=BrickRed,urlcolor=BrickRed,pdfstartview=FitH]{hyperref}

\renewcommand{\P}[1]{{\mathbb{P}}\left[#1\right]}

\newcommand{\I}[1]{{\mathbb{ I}}\left[#1\right]}
\newcommand{\E}[1]{{\mathbb{E}}\left[#1\right]}

\providecommand{\norm}[1]{\left\lVert#1\right\rVert}

\declaretheorem[numberwithin=section]{theorem}
\declaretheorem[sibling=theorem]{lemma}
\declaretheorem[sibling=theorem]{proposition}
\declaretheorem[sibling=theorem]{claim}

\declaretheorem[sibling=theorem]{corollary}

\declaretheorem[sibling=theorem]{definition}

\newenvironment{proofof}[1]{{\medbreak\noindent \em Proof of #1.  }}{\hfill\qed\medbreak}

\def\bx{{\bf x}}

\def\b1{{\bf 1}}
\def\eps{{\epsilon}}

\def\cP{{\cal P}}

\def\os{\tilde{s}^*}
\def\ot{\tilde{t}^*}
\def\size{\sigma}
\def\v{\Delta}
\newcommand{\D}[1]{{#1_{\mathcal{D}}}}

\DeclareMathOperator{\CUT}{CUT}

\DeclareMathOperator{\poly}{poly}

\DeclareMathOperator{\trace}{trace}

\begin{document}

\title{A New Regularity Lemma and \\
Faster Approximation Algorithms for Low Threshold Rank Graphs}
\author{Shayan Oveis Gharan\thanks{Department of Management Science and Engineering, Stanford University. Supported by a Stanford Graduate Fellowship. Email:\protect\url{shayan@stanford.edu}.}
\and
Luca Trevisan\thanks{Department of Computer Science, Stanford University. This material is based upon  work supported by the National Science Foundation under grant No.  CCF 1017403.
Email:\protect\url{trevisan@stanford.edu}.}
}

\date{}
\maketitle

\begin{abstract}
Kolla and Tulsiani \cite{KT07,Kolla11} and Arora, Barak and Steurer \cite{ABS10} introduced the technique of {\em subspace enumeration}, which gives approximation algorithms for graph problems such as unique games and small set expansion; the running time of such algorithms is exponential in the {\em threshold-rank} of the graph.

Guruswami and Sinop \cite{GS11,GS12}, and Barak, Raghavendra, and Steurer \cite{BRS11} developed an alternative approach to the design of approximation algorithms for graphs of bounded threshold-rank, based on semidefinite programming
relaxations in the Lassere hierarchy and on novel rounding techniques. These algorithms  are faster than the ones based on subspace enumeration and work on a broad class of problems.

In this paper we develop a third approach to the design of such algorithms. We show, constructively, that graphs of bounded threshold-rank satisfy a {\em weak Szemeredi regularity lemma} analogous to the one proved by Frieze and Kannan \cite{FK99} for dense graphs. The existence of efficient approximation algorithms is then a consequence of the regularity lemma, as shown by Frieze and Kannan.

Applying our method to the Max Cut problem, we devise an algorithm that is faster than all previous algorithms, and is easier to describe and analyze.
\end{abstract}

\maketitle
\section{Introduction}

Kolla and Tulsiani \cite{KT07,Kolla11} and Arora, Barak and Steurer \cite{ABS10} proved that
the Unique Games problem can be approximated efficiently if the adjacency matrix of a graph
associated with the problem has few large eigenvalues; they show that, for every optimal
solution, its indicator vector is close to the subspace spanned by the eigenvectors of the
large eigenvalues, and one can find a solution close to an optimal one by enumerating
an $\epsilon$-net for such a subspace. Such {\em subspace enumeration} algorithm runs
in time exponential in the dimension of the subspace, which is the number of large eigenvalues;
such a parameter is called the {\em threshold rank} of the graph. Arora, Barak and Steurer
show that the subspace enumeration algorithm can approximate other graph problems, in regular graphs, in time
exponential in the threshold rank, including the Uniform Sparsest Cut problem, the Small-Set Expansion
problem and the Max Cut problem\footnote{The subspace enumeration algorithm
does not improve the 0.878 approximation guarantee of Goemans, Williamson \cite{GW95},
but it finds a solution of approximation factor $1-O(\eps)$ if the optimum cuts at least $1-\eps$.
fraction of edges.}.

Barak, Raghavendra and Steurer \cite{BRS11} and Guruswami and Sinop \cite{GS11,GS12,GS13} developed an alternative approach to the design of approximation algorithms running in time exponential in the threshold rank. Their algorithms are based on solving semidefinite programming relaxations from the Lasserre hierarchy and then applying sophisticated rounding schemes. The advantage of this approach
is that it is applicable to a more general class of graph problems and constraint satisfaction problems, that the approximation guarantee has a tighter dependency on the threshold used in the definition of threshold rank and that, in same cases, the algorithms have a running time of $f(k,\epsilon) \cdot  n^{O(1)}$ where $k$ is the threshold rank and $1\pm \epsilon$ is the approximation guarantee, 
instead of the running time of $n^{O(k)}$ which follows from an application of the subspace enumeration 
algorithm for constant $\epsilon$.

In this paper we introduce a third approach to designing algorithms for graphs of bounded
threshold rank, which is based on proving a {\em weak Szemeredi regularity lemma} for such graphs.

The regularity lemma of Szemeredi  \cite{Szemeredi78} states that every dense graph can be well
approximated by the union of a constant number of bipartite complete subgraphs; the constant, however, has a tower-of-exponentials dependency on the quality of approximation. Frieze and Kannan \cite{FK96,FK99} prove what they call a {\em weak regularity lemma}, showing that every dense
graph can be approximated up to an error $\epsilon n^2$ in the cut norm by a linear combination of
$O(1/\epsilon^2)$ cut matrices (a cut matrix is a bipartite complete subgraph) with bounded coefficients. Frieze and Kannan
also show that such an approximation can be constructed ``implicitly'' in time polynomial in  $1/\epsilon$
and that, for a weighted graph which is a linear combination of $\sigma$ cut matrices, several graph problems can be approximated in time $\exp(\tilde O(\sigma)) + \poly(n)$ time. Combining the two facts one has
a $\exp(\poly(1/\epsilon)) + \poly(n)$ time approximation algorithm for many graph problems on dense graphs. 

We prove that a weak regularity lemma holds for all graphs of bounded threshold rank. Our result is a proper generalization of the weak regularity lemma of Frieze and Kannan, because dense graphs are known to have bounded threshold rank\footnote{The normalization
used for dense graphs is different.}.
For a (weighted) $G = (V,E)$ with adjacency matrix $A$, and diagonal matrix of vertex degrees $D$, $D^{-1/2}AD^{-1/2}$ is called the normalized adjacency matrix of $G$. If the square sum of the eigenvalues of the normalized adjacency matrix  outside the range $[-\epsilon/2,\epsilon/2]$ is equal to $k$ (in particular, if there are at most $k$ such eigenvalues), then
we show that there is a linear combination of $O(k/\epsilon^2)$ cut matrices that approximate $A$
up to $2\epsilon |E|$ in cut norm; furthermore, such a decomposition can be found in $\poly(n,k,1/\eps)$ time. (See \autoref{thm:matrixdec} below.)  Our regularity lemma, combined with an improvement of the
Frieze-Kannan approximation algorithm for graphs that are linear combination of cut matrices, gives
us algorithms of running time $2^{\tilde O(k^{1.5}/\epsilon^3)} + \poly(n)$ for several graph problems
on graphs of  threshold rank $k$, providing an additive approximation of $2\epsilon |E|$. In problems such as Max Cut in which the optimum is $\Omega(|E|)$, this additive approximation is equivalent to
a multiplicative approximation.

\begin{table}[htb]
\begin{tabular}{|c|c|c|}
\hline
Reference & Running time & Parameter $k$\\
\cite{BRS11}* & $2^{O(k/\epsilon^4)}\cdot \poly(n) $ & \# of eigenvalues not in range $[-c\cdot\epsilon^2,c \cdot \epsilon^2]$, $c>0$\\
\cite{GS11}  &  $n^{O(k/\epsilon^2)} $ & \# of eigenvalues $\leq - \epsilon/2$\\
\cite{GS12} & $2^{k/\epsilon^3} \cdot n^{O(1/\epsilon)}$ & \# of eigenvalues $\leq - \epsilon^2/2$\\
\hline
%\mbox{} & \mbox{} & \\
this paper & $2^{\tilde O(k^{1.5}/\epsilon^3)} +  \poly(n)$ & sum of squares of eigenvalues not in range $[-\epsilon/8,\epsilon/8]$\\
\hline
\end{tabular}
\caption{A comparison between previous algorithms applied to Max Cut and our algorithm. \cite{BRS11} needs to solve $r$ rounds of Lasserre hierarchy, for $r=O(k/\eps^4)$.}
\label{tab:comp}
\end{table}

\autoref{tab:comp} gives a comparison between previous algorithms applied to Max Cut and our algorithm.
The advantages over previous algorithms, besides the simplicity of the algorithm, is a faster
running time and the dependency on a potentially smaller theshold-rank parameter, because
the running time of our algorithm depends on the {\em sum of squares} of eigenvalues outside
of a certain range, rather than the number of such eigenvalues. (recall that the eigenvalues of $D^{-1/2}AD^{-1/2}$
 are in the range $[-1,1]$.)

We now give a precise statement of our results, after introducing some notation.

\section{Statement of Results}

\subsection{Notations}
Let $G=(V,E)$ be a (weighted) undirected graph with $n:=|V|$ vertices. 
Let $A$ be the adjacency matrix of $G$.
For a vertex $u\in V$, let $d(u):=\sum_{v\sim u} A(u,v)$ be the degree of $u$. For a set $S\subset V$, let $d(S)=\sum_{v\in S} d(v)$ be the summation of vertex degrees in $S$, and and let $m:=2|E|=d(V)$.
Let $D$ be the diagonal matrix of degrees. 
For any matrix $M\in\mathbb{R}^{V\times V}$, we use $\D{M}$ to denote the symmetric matrix $D^{-1/2}MD^{-1/2}$.
We call $\D{A}$ the {\em normalized adjacency matrix} of $G$.
It is straightforward to see that all eigenvalues of $\D{A}$ is contained in the interval $[-1,1]$.
%We want to
%approximate $A$ with respect to the cut norm within an additive error of $\eps m$. % where $\bar{d} = \frac{1}{n}\sum_{i=1}^n d(i)$ is the average degree of the vertices of $G$.

 For two functions $f,g\in\mathbb{R}^V$, let $\langle f,g\rangle :=\sum_{v\in V} f(v)g(v)$.
 Also, let $f\otimes g$ be the tensor product of $f,g$; i.e., the matrix in $\mathbb{R}^{V\times V}$ such that $(u,v)$ entry is $f(u)g(v)$. For a function $f\in\mathbb{R}^V$, and $S\subseteq V$
 let $f(S):=\sum_{v\in S} f(v)$.

For a set $S \subseteq V$, let $\b1_S$ be the indicator function of $S$, and let
$$ d_S(v):=\begin{cases}
d(v) & v\in S\\
0 & \text{otherwise.}
\end{cases}$$
For any two sets $S,T\subseteq V$, and $\alpha\in\mathbb{R}$,  we  use the notation $\CUT(S,T,\alpha):=\alpha d_S\otimes d_T$ to denote the matrix corresponding to the cut $(S,T)$, where    $(u,v)$ entry of the matrix is $\alpha d(u) d(v)$ if $u\in S, v\in T$ and zero otherwise. We remark that $\CUT(S,T,\alpha)$ is not necessarily a symmetric matrix. 

\begin{definition}[Matrix Norms]
For a matrix $M \in \mathbb{R}^{V\times V}$, and $S,T\subseteq V$, let $$M(S,T):=\sum_{u\in S,v\in T} M_{u,v}.$$
The Frobenius norm and the  cut norm are defined as follows:
\begin{eqnarray*}
\norm{M}_F &:=& \sqrt{\sum_{u,v} M^2_{u,v}}, \\
\norm{M}_C &:=& \max_{S,T\subseteq V} \left|M(S,T)\right|
\end{eqnarray*}
\end{definition}
\begin{definition}[Sum-Squares Threshold Rank]
For  any unweighted graph $G$, with normalized adjacency matrix $\D{A}$, let $\lambda_1,\ldots,\lambda_n$ be the eigenvalues of $\D{A}$  with the corresponding eigenfunctions $f_1,\ldots,f_n$. For $\delta>0$, the $\delta$ sum-squares threshold rank of $A$ is defined as
$$t_\delta(\D{A}) := \sum_{i: |\lambda_i| > \delta} \lambda_i^2.$$
Also, the $\delta$ threshold approximation of $A$ is defined as,
$$ T_\delta(\D{A}):=\sum_{i: |\lambda_i| > \delta} \lambda_i f_i \otimes f_i. $$
\end{definition}

\subsection{Matrix Decomposition Theorem}
The following matrix decomposition theorem is the main technical result of this paper. 
\begin{theorem}
\label{thm:matrixdec}
For any graph $G$, and $\eps>0$, let $k:=t_{\eps/2}(\D{A})$. There is a  algorithm that writes $A$ as a linear combination of  cut matrices, $W^{(1)}, W^{(2)},\ldots, W^{(\size)}$, such that $\size \leq 16k/\eps^2$, and 
$$\norm{A-W^{(1)}-\ldots-W^{(\size)}}_C \leq \eps m,$$
where each $W^{(i)}$ is a cut matrix $\CUT(S,T,\alpha)$, for some $S,T\subseteq V$, such that $|\alpha|\leq  \sqrt{k}/m$. %and $d(S),d(T)\geq  \eps^2m/(16k)$.
The running time of the algorithm is polynomial in $n,k,1/\eps$. 
\end{theorem}

\subsection{Algorithmic Applications}
Our main algorithmic application of \autoref{thm:matrixdec} is the following theorem that approximates any cut on low threshold rank graphs with a running time $2^{\tilde{O}(k^{1.5}/\eps^3)} + \poly(n)$. 
\begin{theorem}
\label{thm:maxcut}
Let $G=(V,E)$, and for a given $0<\eps$, let $k:=t_{\eps/8}(\D{A})$. There is a randomized algorithm such that for any maximization or minimization problem on sets of size $\Gamma$  in time $2^{\tilde{O}(k^{1.5}/\eps^3)}+\poly(n)$ finds a  random set $S$ such that $|d(S)-\Gamma|\leq \eps m$, and with constant probability for any $S^*$ of size $d(S^*)=\Gamma$,
$$ A(S,\overline{S}) \geq A(S^*,\overline{S^*}) - \eps m$$ 
if it is a maximization problem, otherwise,
$$ A(S,\overline{S}) \leq A(S^*,\overline{S^*}) + \eps m.$$
\end{theorem}

We can use the above theorem to provide a  PTAS for maximum cut, maximum bisection, and minimum bisection problems.
\begin{corollary}
Let $G=(V,E)$, and for a given $\eps>0$, let $k:=t_{\eps/8}(\D{A})$. There is a randomized algorithm that in time $2^{\tilde{O}(k^{1.5}/\eps^3)}+\poly(n)$ finds an $\eps m$ additive approximation of the maximum cut.
\end{corollary}
\begin{proof}
We can simply guess the size of the optimum within an $\eps m/2$ additive error and then use \autoref{thm:maxcut}. 
\end{proof}

\begin{corollary}
Let $G=(V,E)$, and for a given $\eps>0$, let $k:=t_{\eps/8}(\D{A})$. For any of the maximum bisection and minimum bisection problems, there is a randomized algorithm that in time $2^{\tilde{O}(k^{1.5}/\eps^3)}+\poly(n)$ finds a cut $(S,\overline{S})$ such that $|d(S)-m/2|\leq \eps m$ and that $A(S,\overline{S})$ provides an $\eps m$ additive approximation of the optimum.
\end{corollary}
\begin{proof}
For the maximum/minimum bisection the optimum must have size $m/2$. So we can simply use \autoref{thm:maxcut} with $\Gamma=m/2$. 
\end{proof}

\section{Regularity Lemma for Low Threshold Rank Graphs}
In this section we prove \autoref{thm:matrixdec}.
The first step is to approximate $A$ by a low  rank matrix $B$. In the next
lemma we construct $B$ such that the value of any cut in $A$ is approximated within
an small additive error in $B$. 

\begin{lemma}
Let $A$ be the adjacency matrix of $G$. For  $0\leq \delta<1$, let
\begin{equation*}
B:= D^{1/2}T_\delta(\D{A})D^{1/2}.
\end{equation*} Then, $\norm{A-B}_C \leq \delta m$.
\label{lem:ABnorm}
\end{lemma}
\begin{proof}
Let $\lambda_1,\ldots,\lambda_n$ be the eigenvalues of $\D{A}$, with the corresponding eigenfunctions $f_1,\ldots, f_n$. For any $S,T\subseteq V$, we have
\begin{eqnarray*}
\langle \b1_S, (A-B) \b1_T\rangle &=& \langle D^{1/2}\b1_S,   \D{(A - B)} D^{1/2} \b1_T\rangle\\
&=& \langle \sqrt{d_S}, (\D{A} - T_\delta(\D{A})) \sqrt{d_T}\rangle\\
%&=& \sqrt{\bd_S^T} \left( \sum_{i: |\lambda_i|\leq \delta} \lambda_i v_i v_i^T \right) \sqrt{\bd_T}\\
&\leq & \delta \sum_{i: |\lambda_i| \leq \delta}  \langle \sqrt{d_S}, f_i\rangle \langle \sqrt{d_T},f_i\rangle \\
&\leq & \delta \sqrt{\sum_{i: |\lambda_i|\leq \delta} \langle \sqrt{d_S},\bx_i\rangle^2}\cdot\sqrt{ \sum_{i:|\lambda_i|\leq \delta} \langle \sqrt{d_T},f_i\rangle^2}\\
&\leq &\delta \norm{\sqrt{d_S}} \norm{\sqrt{d_T}} \leq \delta \norm{\sqrt{d_V}}^2 = \delta m,
\end{eqnarray*}
where the second inequality follows by the Cauchy-Schwarz inequality. 
The lemma follows by noting the fact that $\norm{A-B}_C$ is the maximum of the above expression for any $S,T\subseteq V$.
\end{proof}

%In the rest of the proof we let $B:=D^{1/2} T_\delta(D^{-1/2} A D^{-1/2}) D^{1/2}$.
By the above lemma if we approximate $B$ by a linear combination of cut matrices, that also is a good approximation of $A$. Moreover, since $t_\delta(\D{A}) = t_\delta(\D{B})$, $B$ has a small sum-square threshold rank iff $A$ has a small threshold rank.

\begin{lemma}
For any graph $G$ with adjacency matrix $A$, and $\delta>0$, let $B:=D^{1/2} T_\delta(\D{A})D^{1/2}$. Then,
$$ \norm{\D{B}}_F^2 =  t_\delta(\D{A}).$$
\label{lem:froBnorm}
\end{lemma}
 \begin{proof}
 The lemma follows from the fact that the frobenius norm of any matrix is equal to the summation
 of square of eigenvalues. If $\lambda_1,\ldots,\lambda_n$ are the eigenvalues of $\D{A}$, then
$$ \norm{B}_F^2 = \trace{B^2} = \sum_{|\lambda_i| >\delta} \lambda_i^2 = t_\delta(\D{A})\,.
\qedhere
$$
 \end{proof}

The next proposition is the main technical part of the proof of \autoref{thm:matrixdec}.
We show that we can write any (not necessarily symmetric) matrix $B$
as a linear combination of $O(\norm{B}_F^2/\eps^2)$ cut matrices such that
the cut norm of $B$ is preserved within an additive error of $\eps m$.
The proof builds on the existential theorem of Frieze and Kanan \cite[Theorem 7]{FK99}.

 \begin{proposition}
 \label{lem:cutdecexistence}
For any matrix $B\in \mathbb{R}^{V\times V}$,   $k=\norm{\D{B}}_F^2$, and $\eps>0$, there  exist cut matrices $W^{(1)}, W^{(2)},\ldots,W^{(\size)}$, such
 that $\size\leq 1/\eps^2$, and for all $S,T\subseteq V$,
  $$\left| \left(B-W^{(1)}-W^{(2)}-\ldots-W^{(\size)}\right)(S,T)\right| \leq \eps \sqrt{k} \sqrt{d(S)d(T)},$$  where each $W^{(i)}$ is a cut matrix $\CUT(S,T,\alpha)$, for some $S,T\subseteq V$, and $\alpha\in \mathbb{R}$. 
 \end{proposition}
 \begin{proof}
 \def\X{R}
 Let $\X^{(0)}=B$. We use the potential function $h(\X):=\norm{\D{\X}}_F$.
We show that while $\norm{\X}_C > \eps\sqrt{k} m$, we can choose cut matrices iteratively while maintaining the invariant
that each time the value of the potential function decreases by at
least $\eps^2h(B)$. Since $h(R^{(0)}) = h(B)$,  after at most $1/\eps^2$ we obtain a good approximation of $B$. 

Assume that after $t<1/\eps^2$ iterations, $R^{(i)}=B-W^{(1)}-\ldots- W^{(i)}$. Suppose for some $S,T\subseteq V$, 
\begin{equation}
\label{eq:normdeccond} \left|\X^{(i)}(S,T)\right| >  \eps
h(B)\sqrt{d(S) d(T)}.
\end{equation}
%We show that we can add one more cut matrix and decrease the value of the potential function by at least $\eps^2h(B)^2$. This shows that after at most $\leq 1/\eps^2$ iterations \eqref{eq:normdeccond} does not hold for any $S,T\subseteq V$.
Choose $W^{(i+1)}=\CUT(S, T,\alpha)$,  for $\alpha=\X^{(i)}(S,T) \big/d(S) d(T)$, and let $R^{(i+1)} = R^{(i)} - W^{(i+1)}$. We have,
\begin{eqnarray*}
h(\X^{(i+1)}) - h(\X^{(i)}) = \sum_{u\in S, v\in T} \frac{(\X^{(i)}_{u,v}-\alpha d(u) d(v))^2-{\X^{(i)}_{u,v}}^2}{d(u)d(v)}
&=& -2\alpha\X^{(i)}(S,T) + \alpha^2 d(S)d(T)\\
&=& \frac{-\X^{(i)}(S,T)^2}{d(S)d(T)}  \leq  -\eps^2\Phi^2(B),
\end{eqnarray*}
where the second to last equation follows from the definition of $\alpha$, and the last equation follows from equation \eqref{eq:normdeccond}.
Therefore, after at most $\size\leq 1/\eps^2$ iterations, \eqref{eq:normdeccond} must
hold for all $S,T\subseteq V$. 
%Therefore,
%$$ \norm{B-W^{(1)}-\ldots-W^{(\size)}}_C = \norm{\X^{(\size)}}_C \leq \eps h(B) m \leq  \eps \norm{\D{B}}_F m = \eps \sqrt{k} m.\qedhere$$
\end{proof}

Although the previous proposition only proves the existence of a decomposition into cut matrices, we can construct such a decomposition efficiently using the following nice result of Alon and Naor~\cite{AN06} that gives a consant factor approximation algorithm for the cut norm of any matrix. 
\begin{theorem}[Alon and Naor \cite{AN06}]
\label{thm:cutnormapp}
There is a polynomial time randomized algorithm such that for any given $A\in \mathbb{R}^{V\times V}$, with high probability,
finds sets $S, T \subseteq V$, such that 
$$ |A(S,T)| \geq 0.56 \norm{A}_C. $$
\end{theorem}

\noindent Now we are ready to prove \autoref{thm:matrixdec}.

\begin{proofof}{\autoref{thm:matrixdec}}
Let $\delta:=\eps/2$, and $B:=D^{1/2} T_\delta(\D{A})D^{1/2}$. By \autoref{lem:ABnorm}, we have that $\norm{A-B}_C \leq \eps m/2$. So we just need to approximate $B$ by a set of cut matrices within an additive error of $\eps m/2$. For a matrix $R$, let $h(R):=\norm{\D{R}}_F$. By \autoref{lem:froBnorm} we have $h(B) = \sqrt{k}$. 

Let $\eps':=\eps/\sqrt{4k}$. We use the proof strategy of \autoref{lem:cutdecexistence}. Let $R^{(i)} = B-W^{(1)}-\ldots-W^{(i)}$. If $\norm{R^{(i)}}_C \geq \eps' \sqrt{k} m$, then by \autoref{thm:cutnormapp} in polynomial time we can find $S,T\subseteq V$ such that
\begin{equation}
\label{eq:rstlowerbound}
 \left| R^{(i)}(S,T)\right| \geq \eps' \sqrt{k} m/2 \geq \eps' h(B) m/2. 
 \end{equation}
Choose $W^{(i+1)}=\CUT(S,T,\alpha)$, for $\alpha = R^{(i)}(S,T)/m^2$, and let $R^{(i+1)} = R^{(i)} - W^{(i+1)}$. We get,
$$ h(R^{(i+1)})-h(R^{(i)}) = -2\alpha R^{(i)}(S,T) + \alpha^2 d(S) d(T) \leq -\frac{R^{(i)}(S,T)^2}{m^2} \leq -\frac{\eps'^2 \Phi^2(B)}4.$$
Since $h(R^{(0)})=h(B)$, after $\size\leq 4/\eps'^2 = 16k/\eps^2$, we have $\norm{R^{(\size)}}_C\leq \eps'\sqrt{k}m$, which implies that 
$$\norm{A-W^{(1)}-\ldots-W^{(\size)}}_C\leq \norm{A-B}_C + \norm{B-W^{(1)}-\ldots-W^{(\size)}}_C \leq \eps  m.$$
This proves the correctness of the algorithm. It remains to upper bound $\alpha$.
For each cut matrix $W^{(i)}=\CUT(S,T,\alpha)$ constructed throughout the algorithm we have
\begin{eqnarray} 
|\alpha|=\frac{|R^{(i)}(S,T)|}{m^2} = \frac{1}{m^2}\left|\sum_{u\in S,v\in T} R^{(i)}_{u,v}\frac{\sqrt{d(u)d(v)}}{\sqrt{d(u)d(v)}}\right| &\leq& \frac{1}{m^2} \sqrt{\sum_{u\in S,v\in T} \frac{{R^{(i)}_{u,v}}^2}{d(u)d(v)}}\sqrt{d(S)d(T)} \label{eq:rstupperbound}\nonumber\\
&\leq& \frac{h(R^{(i)})}{m} \leq \frac{h(B)}{m} = \frac{\sqrt{k}}{m}.\nonumber
\end{eqnarray}
where the first inequality follows by the Cauchy-Schwarz inequality, the second
inequality uses $d(S),d(T)\leq m$, 
and the last  inequality follows by the fact that the potential function is decreasing
throughout the algorithm.
This completes the proof of theorem. 
%Furthermore, by equation \eqref{eq:rstlowerbound} we have,
%\begin{eqnarray*}
%\frac{\eps m}{4} \leq |R^{(i)}(S,T)| \leq h(B) \sqrt{d(S)d(T)} \leq  \sqrt{km d(S)}.
%\end{eqnarray*}
%where the second inequality holds by equation \eqref{eq:rstupperbound}.
%Hence, $d(S)\geq \eps^2 m/(16k)$. Similarly, we may obtain $d(T)\geq \eps^2m/(16k)$.
 \end{proofof}
 
 \section{Fast Approximation Algorithm for Low Threshold Rank Graphs}
\def\sb{\beta}
In this section we prove \autoref{thm:maxcut}. 
First, by \autoref{thm:matrixdec} in time $\poly(n, 1/\eps)$ we can find cut matrices $W^{(1)},\ldots,W^{(\size)}$ for $\size=O(k/\eps^2)$, such that for all $1\leq i\leq t$, $W^{(i)}=\CUT(S_i, T_i, \alpha_i)$,  $\alpha_i\leq \sqrt{k}/m$, and 
\begin{equation*}
 \norm{A-W}_C \leq \eps m/4,
 \end{equation*}
where $W:=W^{(1)}+\ldots+W^{(\sigma)}$. It follows from the above equation that for any set $S\subseteq V$,
\begin{equation}
\label{eq:cutapp}
|A(S,\overline{S}) - W(S,\overline{S})| = |A(S,\overline{S}) - \sum_{i=1}^\size \alpha_i d(S\cap S_i) d(\overline{S}\cap T_i)| \leq \frac{\eps m}{4}.
\end{equation}
Fix $S^*\subseteq V$ of size $d(S^*)=\Gamma$ (think of $(S^*,\overline{S^*})$ as the optimum cut), and let $s^*_i:=d(S_i\cap S^*)$, and $t^*_i:=d(T_i\cap \overline{S^*})$. Observe that by equation \eqref{eq:cutapp},
\begin{equation}
\label{eq:sitiapp}
\left| A(S^*,\overline{S^*}) - \sum_{i=1}^\size \alpha_i s^*_i t^*_i\right| \leq \frac{\eps m}{4}.
%\sum_{i=1}^\size W^{(i)}(S^*,\overline{S^*}) = \sum_{i=1}^\size \sum_{\substack{u\in S_i\cap S^*\\ v\in T_i\cap \overline{S^*}}} d(u)d(v)\alpha_i=\sum_{i=1}^\size s^*_i t^*_i \alpha_i.
\end{equation} 
Let $\alpha_{\max}:=\max_{1\leq i\leq \size} |\alpha_i|$. 
Let $\v := \lfloor \eps /(48\alpha_{\max}\size)  \rfloor$; observe that $\Delta=O(\eps^3m/k^{1.5})$. We define an approximation of $s^*_i,t^*_i$ by rounding them down to the nearest multiple of $\v$, i.e.,   $\os_i=\v \lfloor s^*_i/\v\rfloor$, and $\ot_i = \v \lfloor t^*_i/\v\rfloor$. We use $\os,\ot$ to denote the vectors of the approximate values. It follows that we can obtain a good approximation of the size of the cut $(S^*,\overline{S^*})$ just by guessing  the vectors $\os$,  and $\ot$. Since $|s^*_i-\os_i|\leq \v$ and $|t^*_i-\ot_i|\leq \v$, we get,
\begin{equation}
\label{eq:osiotiapp}
\sum_{i=1}^\size |s^*_it^*_i\alpha_i - \os_i\ot_i\alpha_i|  \leq \size\alpha_{\max} (2\v m + \v^2) \leq 3\alpha_{\max}\size\Delta \leq \eps m /16.
\end{equation}

Observe that by equations \eqref{eq:cutapp},\eqref{eq:sitiapp},\eqref{eq:osiotiapp}, if we know the vectors $\os,\ot$, then we can find $A(S^*,\overline{S^*})$ within an additive error of $\eps m/2$. 
Since $\os_i,\ot_i\leq m$, there are only $O(m/\v)$ possibilities for each $\os_i$ and $\ot_i$. Therefore, we afford to enumerate all possible values of them in time $(m/\v)^{2\size}$, and choose the one that gives the largest cut. Unfortunately, for a given assignment of $\os,\ot$ the corresponding cut $(S^*,\overline{S^*})$ may not exist.
Next we give an algorithm that for a given assignment of $\os,\ot$ finds a cut $(S,\overline{S})$ such that $A(S,\overline{S}) = \sum_i \os_i\ot_i\alpha_i\pm \eps m$, if one exists.

First we distinguish the large degree vertices of $G$ and simply guess which side they are mapped to in the optimum cut. For the rest of the vertices we use the solution of LP(1). Let $U:=\{v: d(v)\geq \Delta\}$ be the set of large degree vertices. Observe that $|U|\leq m/\Delta$. Let $\cP$ be the coarsest partition of the set $V\setminus U$ such that for any $1\leq i\leq \size$, both $S_i\setminus U$ and $T_i\setminus U$ can be written as a union of sets in $\cP$, and for each $P\in \cP$, $d(P)\leq \Delta$.   Observe that $|\cP|\leq 2^{2\size}+m/\Delta$. For a given assignment of $\os,\ot$, first  we guess the set of vertices in $U$ that are contained in $S^*$, $U_{S^*}:=S^*\cap U$, and $U_{\overline{S^*}}:=U\setminus U_{S^*}$.
For the rest of the vertices we use the linear program LP(1) to find the unknown $d(S^*\cap P)$. 
\begin{alignat}{6}
& \textup{LP(1)} \quad&& \quad&& \quad&& \quad&& \quad&& \nonumber \\
 & 0 \quad&& \leq \quad&& y_P \quad&& \leq \quad&& 1  \quad&& \forall P\in \cP \nonumber\\ 
& \Gamma-\eps m/2 \quad&& \leq  \quad&& \sum_{P} y_P d(P)+d(U_{S^*}) \quad&& \leq \quad&& \Gamma + \eps m/2 \quad&&           \label{eq:yexpsize}\\
& \os_i \quad && \leq \quad&& \sum_{P\subseteq S_i} y_Pd(P)+d(U_{S^*}\cap S_i) \quad&& \leq \quad&& \os_i+\v \quad&&  \forall 1\leq i\leq \size \label{eq:yosi}\\
& \ot_i \quad&& \leq \quad&& \sum_{P\subseteq T_i} (1-y_P)d(P)+d(U_{\overline{S^*}}\cap T_i) \quad&& \leq \quad&& \ot_i+\v \quad&& \forall 1\leq i\leq \size. \label{eq:yoti}
\end{alignat}
Observe that $y_P=\frac{d(S^*\cap P)}{d(P)}$ is a feasible solution to the linear program.
In the next lemma which is the main technical part of the analysis we show how to construct a set based on a given solution of the LP. 
\begin{lemma}
\label{lem:lprounding}
There is a randomized algorithm such that for any $S^*\subset V$, given $\os_i,\ot_i$ and $U_{S^*}$  returns a random set $S$ such that 
% $$ |\sum_{i=1}^\size d(S^*\cap S_i)d(\overline{S^*}\cap T_i)\alpha_i - \os_i\ot_i\alpha_i | \leq \eps m /2.$$
%\begin{alignat}{3}
%& \P{|d(S) - d(S^*)| \geq \eps m} \quad&&\leq \quad&& \frac{\eps}{8},\label{eq:expectedsize}\\
%& \left|\E{W(S,\overline{S})} -  A(S^*,\overline{S^*})\right| \quad&&\leq \quad&&  \frac{\eps m}2.
% \label{eq:expectedcut}
% \end{alignat}
\begin{eqnarray}
 \P{W(S,\overline{S}) \geq A(S^*,\overline{S^*})-\frac{3\eps m}4 \wedge |d(S) - \Gamma|\leq \eps m}&\geq &\frac{\eps}{10} \label{eq:wlarger}\\
  \P{W(S,\overline{S}) \leq A(S^*,\overline{S^*})+\frac{3\eps m}4 \wedge |d(S) - \Gamma|\leq \eps m}&\geq& \frac{\eps}{10}. \label{eq:wsmaller} 
\end{eqnarray}
%Moreover, such a set can be found in with high probability in time $\poly(n,1/\eps)$. 
\end{lemma}
 \begin{proof}
 Let $y$ be a feasible solution of LP(1).
We use a simple independent rounding scheme to compute the random set $S$. 
We always include $U_{S^*}$ in $S$.  For each $P\in \cP$, we include $P$ in $S$, independently, with probability $y_P$. We prove that $S$ satisfies lemma's statements.
First of all, by linearity of expectation,
$$  \E{d(S \cap S_i)} = d(U_{S^*})+\sum_{P\subseteq S_i} y_P d(P), \text{ ~~and~~ } \E{d(\overline{S}\cap T_i)} = d(U_{\overline{S^*}})+\sum_{P\subseteq T_i} (1-y_P)d(P). $$

In the following two claims, first we show that with high probability the expected size of $d(S)$ 
is close to $\Gamma$. Then, we upper bound the expected value of $W(S,\overline{S})-A(S^*,\overline{S^*})$.
\begin{claim}
\label{cl:sizecloseopt}
$$ \P{|d(S) - d(S^*)| \geq \eps m} \leq  \frac{\eps}{8},\label{eq:expectedsize}$$
\end{claim}
\begin{proof}
We use the theorem of Hoeffding to prove the claim:
\begin{theorem}[Hoeffding Inequality]
\label{thm:hoeffding}
Let $X_1, \ldots, X_n$ be independent random variables such that for each $1\leq i\leq n$, $X_i\in [0,a_i]$. Let $X:=\sum_{i=1}^n X_i.$ Then, for any $\eps>0$
$$ \P{|X-\E{X}| \geq \eps} \leq 2\exp\left(-\frac{2\eps^2}{\sum_{i=1}^n a_i^2}\right).$$
\end{theorem}

Now, by the independent rounding procedure, we  obtain
\begin{eqnarray*}
 \P{|d(S) - \E{d(S)}| \geq \eps m/2} \leq 2\exp\left(-\frac{\eps^2m^2}{2\sum_P d(P)^2}\right)
\leq  2\exp\left(-\frac{\eps^2 m^2}{2m\Delta}\right) \leq 2\exp(-16/\eps) \leq \frac{\eps}{8}.
\end{eqnarray*}
where the third inequality follows by the fact that $d(P)\leq \Delta$ and $\sum_P d(P)\leq m$. 
The claim follows from the fact that  by \eqref{eq:yexpsize}, $|\E{d(S)} - \Gamma| \leq \eps m/2$.
\end{proof}

\begin{claim}
\label{cl:cutcloseopt}
$$ \left|\E{W(S,\overline{S})} -  A(S^*,\overline{S^*})\right| \leq \frac{\eps m}2.
$$
\end{claim}
\begin{proof}
First, observe that
\begin{eqnarray}
%\EE{R\sim \cD}{A(S,\overline{S})} &\geq& 
%\E{ \sum_{i=1}^{\size} d(S\cap S_i) d(\overline{S} \cap T_i) \alpha_i}  
\E{W(S,\overline{S})}&=&\E{ \sum_{i=1}^{\size} d(S\cap S_i) d(\overline{S} \cap T_i) \alpha_i}  \nonumber\\
%&=& \sum_{i=1}^\size \EE{R\sim \cD}{d(S\cap S_i) d(\overline{S} \cap T_i)\alpha_i} - \eps m/8\\
&=& \sum_{i=1}^\size \alpha_i \E{(\sum_{P\in\cP:P\subseteq S_i} d(P)\I{P\subseteq S})(  \sum_{Q\in\cP:Q\subseteq T_i}d(Q) \I{Q\subseteq \overline{S}})}\nonumber\\
&&+\sum_{i=1}^\size \alpha_i [d(U_{S^*}\cap S_i)\E{d(\overline{S}\cap T_i)}+ d(U_{\overline{S^*}}\cap T_i)\E{d(S\cap S_i)}] \nonumber\\
&=& \sum_{i=1}^\sigma \alpha_i \sum_{P\subseteq S_i,Q\subseteq T_i} d(P) d(Q) \E{\I{P\subseteq S}\I{Q\subseteq \overline{S}}}\nonumber\\
&&+\sum_{i=1}^\size \alpha_i \left\{d(U_{S^*}\cap S_i)t_i+d(U_{\overline{S^*}}\cap T_i)s_i\right\}. 
\label{eq:wmadesimple}
\end{eqnarray}
Since the event that $P\subseteq S$ is independent of $Q\subseteq S$, iff $P\neq Q$ we get
$$\E{\I{P\subseteq S}\I{Q\subseteq \overline{S}}} = \begin{cases}
y_P (1-y_Q) & \text{if} P\neq Q\\
0 & \text{otherwise}.
\end{cases}$$
Let $s_i:=\E{d(S\cap S_i)}$ and $t_i:=\E{d(\overline{S}\cap T_i)}$. 
Then, by \eqref{eq:wmadesimple} and above equation,
\begin{eqnarray}
\E{ W(S,\overline{S})}  &=& \sum_{i=1}^\size \alpha_i s_i t_i - \sum_{i=1}^\size \alpha_i\sum_{P\in \cP} y_P (1-y_P)d(P)^2.
\label{eq:expectedcutequality}
\end{eqnarray}
On the other hand, by equations \eqref{eq:yosi} and \eqref{eq:yoti}, for all $1\leq i\leq \size$, we get  $\os_i \leq s_i \leq \os_i+\v$ and  $\ot_i\leq t_i\leq \ot_i+\v$. Hence, similar to equation \eqref{eq:osiotiapp} we can show,
\begin{equation}
\label{eq:stsclose}
\sum_{i=1}^\size |\alpha_i s_it_i -  \alpha_i\os_i\ot_i| \leq \frac{\eps m}{8}.
\end{equation}

\noindent Therefore, using equation~\eqref{eq:sitiapp} we get
\begin{eqnarray*}
\left|\E{W(S,\overline{S})} - A(S^*,\overline{S^*})\right| &\leq & \Big| \E{ W(S,\overline{S})} - \sum_{i=1}^\size \alpha_i s^*_i t^*_i)\Big|+\frac{\eps m}{4}\\
&= & \Big| \sum_{i=1}^\size (\alpha_i s_i t_i -\alpha_i s^*_it^*_i) -\sum_{i=1}^\size\sum_{P\in\cP} \alpha_i y_P(1-y_P)d(P)^2\Big| +\frac{\eps m}{4}\\
&\leq & \sum_{i=1}^\size |\alpha_i s_it_i -  \alpha_i\os_i\ot_i| + \sum_{i=1}^\size |\alpha_i\os_i\ot_i - \alpha_i s^*_i t^*_i| + \sigma \alpha_{\max} m\Delta +\frac{ \eps m}{4}
 \leq \frac{\eps m}{2},
\end{eqnarray*}
where the  equality follows by \eqref{eq:expectedcutequality}, the second  inequality follows by the fact that $d(P)\leq \Delta$ for all $P\in \cP$ and $\sum_P d(P) \leq m$, and the last inequality follows by  equations \eqref{eq:stsclose} and  \eqref{eq:osiotiapp}. This proves the claim.
 \end{proof}
 
 Now we are ready finish the proof of \autoref{lem:lprounding}. Here, we prove \eqref{eq:wlarger}. Equation \eqref{eq:wsmaller} can be proved similarly.  By \autoref{cl:cutcloseopt},
 \begin{eqnarray*} A(S^*,\overline{S^*}) -\frac{\eps m}{2} \leq \E{W(S, \overline{S})} &\leq& \E{W(S,\overline{S}) ~|~ |d(S) - \Gamma|\leq \eps m} + m \P{|d(S) - \Gamma| > \eps m} \\
&\leq & \E{W(S,\overline{S}) ~|~ |d(S) - \Gamma|\leq \eps m} + \frac{\eps m}{8}.
\end{eqnarray*}
where the second inequality holds by the fact that the size of any cut in $G$ is at most $m/2$,
thus by \eqref{eq:cutapp} for any $S\subseteq V$, $W(S,\overline{S})\leq \eps m/4+m/2\leq m$,
and the last inequality follows by \autoref{cl:sizecloseopt}. Hence,
\begin{eqnarray*} 
\E{W(S,\overline{S}) ~|~ |d(S) - \Gamma|\leq \eps m} \geq A(S^*,\overline{S^*})-\frac{5\eps m}{8}
\end{eqnarray*}
Since $W(S,\overline{S})\leq m$,
\begin{eqnarray*}
\P{W(S,\overline{S}) \geq A(S^*,\overline{S^*})-\frac{3\eps m}4 ~\Big|~ |d(S) - \Gamma|\leq \eps m}\geq \frac{\eps}{8} 
\end{eqnarray*}
Therefore, \eqref{eq:wlarger} follows by an application of \autoref{cl:sizecloseopt}.
 \end{proof}
 
Our  rounding algorithm is described in Algorithm \ref{alg:setselection}. 
First, we prove the correctness, then we calculate the running time of the algorithm.
Let $S$ be the output set of the algorithm. First, observe that we always have $|d(S) - \Gamma|\leq \eps m$. Now let $A(S^*,\overline{S^*})$ be the maximum cut among all sets of size $\Gamma$ (the minimization case can be proved similarly). In the iteration that the algorithm correctly guesses $\os_i,\ot_i, U_{S^*}$, there exists a feasible solution $y$ of LP(1).
by \autoref{lem:lprounding}, for all $1\leq i\leq 10/\eps$,
$$ \P{W(R_y(i),\overline{R_y(i)}) \geq A(S^*,\overline{S^*})-\frac{3\eps m}4 \wedge |d(R_y(i)) - \Gamma|\leq \eps m}\geq \frac{\eps}{10} 
 $$
Since we take the best of $10/\eps$ samples,   with probability $1/e$ the output set $S$ satisfies
$W(S,\overline{S})\geq A(S^*,\overline{S^*})-3\eps m/4$. Therefore, by \eqref{eq:cutapp},
$A(S,\overline{S})\geq A(S^*,\overline{S^*})-\eps m$.
This proves the correctness of the algorithm.

 \begin{algorithm}
 \begin{algorithmic}
\FORALL {possible values of $\os_i,\ot_i$, and $U_{S^*}\subseteq U$ }
\IF {there is a feasible solution $y$ of LP(1)}
\FOR {$i=1 \to 10/\eps$}
\STATE $R_y(i)\leftarrow U_{S^*}$.
\STATE For each $P\in \cP$ include $P$ in $R_y(i)$, independently, with probability $y_P$.
\ENDFOR
\ENDIF
\ENDFOR
\RETURN among all sets $R_y(i)$ sampled in the loop that satisfy $|d(R_y(i))-\Gamma|\leq \eps m$, the one  that $W(R_y(i),\overline{R_y(i)})$ is the maximum.
%\RETURN $\argmax_{S_y} A(S_y,\overline{S_y})$. % (resp. $\argmin_{S_y} A(S_y,\overline{S_y})$).
 \end{algorithmic}
 \caption{Approximate Maximum Cut $(S,\overline{S})$ such that $d(S)=\Gamma\pm \eps m$}
 \label{alg:setselection}
 \end{algorithm}
 
 It remains to upper-bound the running time of the algorithm. First observe that if $|U|=O(k/\eps^2)$, the running time of the algorithm is dominated by the time it takes to compute a feasible solution of LP(1). Since the size of LP is $2^{\tilde{O}(k/\eps^2)}$, in this case  Algorithm~\ref{alg:setselection} terminates in time $2^{\tilde{O}(k/\eps^2)}$. Note that
 for any sample set $R_y(i)$, both $d(R_y(i))$ and $W(R_y(i),\overline{R_y(i)})$ can be computed
 in time $2^{\tilde{O}(k/\eps^2)}$, once we know $|R_y(i)\cap P|$ for any $P\in \cP$.
 
Otherwise if $|U|\gg k/\eps^2$, the dependency of the running time of the algorithm to $\eps, k$ is dominated by the step where we guess the subset of $U_{S^*} = U\cap S^*$. Since $\alpha_{\max}\leq \sqrt{k}/m$ and $\size = O(k/\eps^2)$, we get 
 $$|U|\leq \frac{m}{\Delta} \leq \frac{12 m\alpha_{\max}\size}{\eps} =O\left(\frac{k^{1.5}}{\eps^3} \right).$$
 Therefore,  Algorithm~\autoref{alg:setselection} runs in time $2^{\tilde{O}(k^{1.5}/\eps^3)}$.
 Since it takes $\poly(n,k,1/\eps)$ to compute the decomposition into $W^{(1)},\ldots,W^{(\sigma)}$, 
 the the total running time is $2^{\tilde{O}(k^{1.5}/\eps^3)}+\poly(n)$. 
 This completes the proof of \autoref{thm:maxcut}.
\bibliographystyle{alpha}
\bibliography{references}
\end{document}